\documentclass[aps,twocolumn,showpacs,pra,citeautoscript,amsmath,amssymb,floatfix,superscriptaddress, bbm, changes,10pt]{revtex4-1}
\pdfoutput=1
\usepackage{amsmath}
\usepackage{amssymb}
\usepackage{epsfig}
\usepackage{color}
\usepackage{amsthm}
\usepackage{hyperref}
\usepackage{appendix}
\usepackage{cancel}
\usepackage{stmaryrd}
\usepackage{soul}
\usepackage{comment}
\usepackage{changes}
\usepackage{empheq}

\newcommand*\widefbox[1]{\fbox{\hspace{2em}#1\hspace{2em}}}
 
\def\<{\langle}
\def\>{\rangle}

\newcommand{\be}{\begin{eqnarray} \begin{aligned}}
\newcommand{\ee}{\end{aligned} \end{eqnarray} }
\newcommand{\benn}{\begin{eqnarray*} \begin{aligned}}
\newcommand{\eenn}{\end{aligned} \end{eqnarray*} }

\newcommand{\ben}{\begin{eqnarray} \begin{aligned}}
\newcommand{\een}{\end{aligned} \end{eqnarray} }
\newcommand\numberthis{\addtocounter{equation}{1}\tag{\theequation}}

\newcommand{\bc}{\begin{center}}
\newcommand{\ec}{\end{center}}

\newcommand{\e}{\mathrm{e}}

\newcommand{\beq}{\begin{eqnarray} \begin{aligned}}
\newcommand{\eeq}{\end{aligned} \end{eqnarray} }
\newcommand{\bea}{\begin{array}}
\newcommand{\eea}{\end{array}}

\newcommand{\bee}{\begin{enumerate}}
\newcommand{\eee}{\end{enumerate}}
\newcommand{\bei}{\begin{itemize}}
\newcommand{\eei}{\end{itemize}}
\newtheorem{theorem}{Theorem}

\newtheorem{lemma}[theorem]{Lemma}

\usepackage{amsfonts}

\def\01{\{0,1\}}

\newcommand{\ket}[1]{|#1\rangle}
\newcommand{\bra}[1]{\langle#1|}

\newcommand{\ketbra}[2]{|#1\rangle\! \langle#2|}

\usepackage{mathtools}

\DeclarePairedDelimiterX{\infdivx}[2]{(}{)}{%
  #1\;\delimsize\|\;#2%
}

\def\<{\langle}
\def\>{\rangle}

\newtheorem*{rep@theorem}{\rep@title}
\newcommand{\newreptheorem}[2]{%
\newenvironment{rep#1}[1]{%
 \def\rep@title{#2 \ref{##1} (restatement)}%
 \begin{rep@theorem}}%
 {\end{rep@theorem}}}
\makeatother

\def\e{\mathrm{e}}
\def\i{\mathrm{i}}

\newreptheorem{thm}{Theorem}
\newreptheorem{lem}{Lemma}

\begin{document}

\title{
Finite-bath corrections to the second law of thermodynamics
}

\author{Jonathan G. Richens}
\email{jonathan.richens08@ic.ac.uk}
\affiliation{Controlled Quantum Dynamics Theory Group, Department of Physics, Imperial College London, London SW7 2AZ, UK.}
\affiliation{Department of Physics and Astronomy, University College London,
Gower Street, London WC1E 6BT, UK.}
\author{\'{A}lvaro M. Alhambra}
\affiliation{Department of Physics and Astronomy, University College London,
Gower Street, London WC1E 6BT, UK.}
\author{Lluis Masanes}
\affiliation{Department of Physics and Astronomy, University College London,
Gower Street, London WC1E 6BT, UK.}

\begin{abstract}
The second law of thermodynamics states that a system in contact with a heat bath can undergo a transformation if and only if its free energy decreases. 
However, the ``if" part of this statement is only true when the effective heat bath is infinite.
In this article we remove this idealization and derive corrections to the second law in the case where the bath has a finite size, or equivalently finite heat capacity. 
This can also be translated to processes lasting a finite time, and we show that thermodynamical reversibility is lost in this regime. 
We do so in full generality, that is without assuming any particular model for the bath;
the only parameters defining the bath are its temperature and heat capacity. 
 We find connections with second order Shannon information theory, in particular in the case of Landuaer erasure. 
We also consider the case of non-fluctuating work, and derive finite-bath corrections to the min and max free energies employed in single-shot thermodynamics.
\end{abstract}
\maketitle

\section{Introduction}


Currently there is an ongoing effort to generalize the laws of thermodynamics to the microscopic regime, motivated by (amongst other things) the development of quantum technologies, the miniaturization of devices and biophysics. 
On the one hand, we are seeing a steady growth in our experimental capabilities
~\cite{blickle2012realization, trotzky2012probing, abah2012single, koski2014experimental, pekola2015towards}, and on the other, new theoretical tools and fomalism are being developed (see~\cite{goold2016role, vinjanampathy2015quantum} for a review). 
In this new regime, some of the standard assumptions of macroscopic thermodynamics are not necessarily valid.
For example, the fluctuations of some quantities, like work, can be much larger than the scale of the system, and hence, can no longer be neglected (a generalization of the second law to the case when fluctuations are bounded was obtained in~\cite{richens2016quantum}). 
Another example is that the environment or heat bath of a microscopic system cannot always be assumed to have infinite size.
In particular, when the time scales involved in a process are small, the effective heat bath is necessarily finite (as implied by Lieb-Robinson bounds). In this work we obtain corrections to the second law in the case that the bath has finite size.


The (macroscopic) second law of thermodynamics can be stated as follows.
Consider a system in an arbitrary  state, with average energy $U$ and entropy $S$.
Suppose that, in order to transform the system, we can make use of a heat bath at temperature $T$, and extract an average amount of work $\langle W \rangle$. 
A necessary condition for the possibility of a transformation is
\begin{equation}
\label{eq:landauer}
  \langle W \rangle 
  \leq 
  -\Delta U + T\Delta S
  =-\Delta F
  \ ,
\end{equation}
where $\Delta U$ and $\Delta S$   are the changes in energy and entropy, and the free energy is defined by $F=U-TS$.
Remarkably, if the available heat bath has infinite heat capacity this condition is also sufficient~\cite{alhambra2016second}.
This also implies that condition~\eqref{eq:landauer} is tight. Hence, it is in principle possible to implement such transformation and extract $\langle W \rangle =-\Delta F$ work.
However, in the microscopic scale, the assumption that the bath is infinite might not always be a reasonable approximation.
Recall that, within a finite time, a system can only interact with a finite region of its environment.
Hence for finite-time processes one should consider an effective bath with a finite heat capacity.

One can ask then, how does this fundamental bound~\eqref{eq:landauer} change when we take into account the finiteness of the heat bath? 
We address this question for the case where the heat capacity of the bath $C$ is not necessarily infinite.
We find that the necessary and sufficient condition for the possibility of a transformation is
\begin{equation} \label{eq:corfree}
  \langle W\rangle 
  \leq 
  -\Delta F -\frac{1}{ \beta}
  D\left[ P_C (s',E') \| P_\infty (s',E') \right]
  \ ,
\end{equation}
where $D$ is the relative entropy, $P_C (s',E')$ is the joint probability distribution for the final state of the system $s'$ and the energy of the bath $E'$ when this has heat capacity $C$ and $P_\infty (s',E')$ is the final distribution when the bath has infinite heat capacity. 
As we will see, this distribution  factorizes as $P_\infty (s',E') = P(s')p_G(E')$, with the final distribution of the bath being the Gibbs or thermal state $p_G (E')$. From~\eqref{eq:corfree} we derive a necessary condition that only the initial and final states of the system, and the heat capacity $C$ and temperature $T=1/\beta$ of the bath
\begin{equation} \label{eq:corfree2}
\langle W\rangle \leq -\Delta F -\frac{1}{2 C\beta}\Delta S^2
\ .
\end{equation}

Note that this expression converges to the macroscopic second law~\eqref{eq:landauer} in the limit where the heat capacity of the bath $C$ is large; and provides a stronger condition when $C$ is finite.
Also, condition~\eqref{eq:corfree2} becomes sufficient when the final state is maximally mixed.
In this case, expression~\eqref{eq:corfree2} is not just an upper bound, but the maximal work extractable.
Not that this is smaller than the minimal work invested in the preparation of the initial state. Therefore, we conclude that thermodynamics reversibility requires an infinite heat bath.



There are situations where stochastic fluctuations of work may be undesirable, and one would like to manipulate only definite amounts of useful energy. To address this, the framework of single-shot thermodynamics was established in \cite{aaberg2013truly,horodecki2013fundamental}, where the concept of \emph{deterministic work} was introduced. The authors showed how work can be understood as a shift of energy in the storage system that happens with a very large probability. They also gave expressions for how much one can extract from any given state, and of how much one needs to create it, assuming access to an infinite heat bath.
%
 
 We here explore how that amount of deterministic work content and work cost of forming a state changes when one is limited to a finite bath and we find that achieving these processes with arbitrary accuracy becomes impossible. This may not the case, however, if one allows for an additional small error probability during the processes, given by the tails of the distribution of energy of the heat bath. For such case, we give expressions for the deterministic work for a given probability of failure. 

Finite bath thermodynamics has garnered some interest in recent years, including finite bath corrections to Carnot efficiency, Landauer's principle and the Jarzinski equality  \cite{woods2015maximum,reeb2014improved,campisi2009finite,scharlau2016quantum}. All of these contributions make assumptions on the on the particular structure of the heat bath.  Our contribution has the advantage of not making any assumption on the structure of the heat bath.

The structure of this paper is as follows: in Section \ref{sec:finite}, we provide a model-independent characterization of a finite bath. 
In Section \ref{sec:TOfluc} describe the general model for thermodynamic system-bath interactions and consider work extraction protocols with fluctuating work. We use this to show our main result and explore the much studied case of Landauer erasure. 
In Section \ref{sec:detTO} we further explore the possible interactions between system and bath only, and how the system can be transformed via operations that do not involve work. We use this to derive fundamental limits to deterministic work extraction and expenditure. All the technical proofs are given in the corresponding appendices.

\section{General characterization of a finite bath}\label{sec:finite}

We consider a bath to be a large (but in this case not infinitely so) system with a density of states given by $\Omega (E,V)=e^{S(E,V)}$, where $S(E,V)$ is the entropy in the microcanonical ensemble for a given energy $E$ and volume $V$. We shall make three assumptions about it
\begin{enumerate}

	\item The entropy $S(E,V)$ is extensive: $S(kE,kV)=kS(E,V)$ for all $k>0$
	\item The dimensionless volume $V$ is large. 
	\item The bath is in a Gibbs state with a given temperature $\beta$, such that a microstate of energy $E$ has probability $\frac 1 Z e^{-\beta E}$.
\end{enumerate}
We will be working in units for which the Boltzmann constant is $k_B=1$. Assumption $1$ implies that we can write the entropy $S(E,V)$ as
\begin{equation}
S(E,V)=V f(u)
\end{equation}
for some function $f(u)$ of the  energy density $u=E/V$. 
The probability distribution for $u$ is
then
\begin{equation}\label{eq:probu}
p(u) \propto e^{V[f(u)-\beta u]}
\ .
\end{equation}
In the large $V$ limit we can use the saddle point approximation \cite{butler2007saddlepoint}
\begin{equation}\label{eq:probu}
  p(u) \propto 
  e^{V \left[ f(u_\beta) +\frac 1 2 f''(u_\beta) (u-u_\beta)^2 \right]}
\ ,
\end{equation}
where $u_\beta$ is the absolute maximum of $f(u)-\beta u$ as a function of $u$, which implies
\begin{eqnarray}
  \label{f' id}
  f' (u_\beta) = \beta\ ,
  \\
  \label{sign f''}
  f'' (u_\beta)<0\ .
\end{eqnarray}
In summary, we have a normal distribution
\begin{equation}
p(u) \propto e^{- \frac V 2 |f''(u_\beta)| (u-u_\beta)^2}
\ ,
\end{equation}
with mean $\langle u\rangle = u_\beta$ and variance $\langle (u -\langle u\rangle)^2 \rangle = |V f''(u_\beta)|^{-1}$.

Now, let us relate $f'' (u_\beta)$ to the heat capacity, defined as
\begin{equation}
  \label{c def}
  C = 
  V \frac{\text{d}\langle u \rangle}{\text{d}T} =
  -\frac V {T^2} \frac{\text{d}\langle u \rangle}{\text{d}\beta}
  \ .
\end{equation}
Differentiating~\eqref{f' id} with respect to $\beta$ gives $f'' (u_\beta) \frac {\text d u_\beta} {\text d \beta} =1$, and  substituting in~\eqref{c def} gives
\begin{equation}
  \label{c def 2}
  C = 
  -\frac V {T^2} \frac 1 {f'' (u_\beta)}
  \ .
\end{equation}
Note that~\eqref{sign f''} implies that the heat capacity is positive, as is always the case in ``ordinary matter".
Also note that $C\propto V$, because $f(u_\beta)$ is independent of $V$.
Which implies that the fluctuations of $u$ are 
\begin{equation}
  \left\langle (u -\langle u\rangle)^2 \right\rangle ^{1/2} 
  = 
  \frac {T \sqrt C} V
  \propto V^{-1/2}
  \ ,
\end{equation}
which are small when $V$ is large.

We are finally able to approximate the density of states of the bath as
\begin{align}\label{eq:densityapp}
  \Omega(E,V) \propto
  \text{exp}\! \left( {\beta E - \frac{\gamma E^2}{2}} \right) ,
\end{align}
where we have rescaled the energy such that $\langle E \rangle =0$, and we define $\gamma = \frac{1}{C T^2}$.

\section{Fluctuating work}\label{sec:TOfluc}

\subsection{Thermal operations with fluctuating work}

Next we introduce a widely use framework to describe thermodynamic transformations \cite{skrzypczyk2014work,masanes2014derivation,gemmer2015single}. 
Our setting consists of a system with Hamiltonian $H_{\rm S}$, the bath with Hamiltonian $H_{\rm B}$ initially in the thermal state (as describe in previous section), and an ideal weight with Hamiltonian $H_\mathrm{W}=  \int_\mathbb{R} dx\, x |x\rangle \! \langle x|$, where the orthonormal basis $\{|x\rangle, \forall\, x\in \mathbb R\}$ represents the position of the weight. Any joint transformation of system, bath and weight is represented by a Completely Positive Trace Preserving (CPTP) map $\Gamma_\mathrm{SBW}$ satisfying the following conditions:
\begin{description}
  \item[Microscopic reversibility (Second Law)] It has an (CPTP) inverse $\Gamma_\mathrm{SBW}^{-1}$, which implies unitarity $\Gamma_\mathrm{SBW} (\rho_\mathrm{SBW}) = U\rho_\mathrm{SBW} U^\dagger$.
    
  \item[Energy conservation (First Law)] 
  $[U,H_\mathrm{S} +H_\mathrm{B} +H_\mathrm{W}] =0$. 

  \item[Independence from the ``position" of the weight]
  The unitary commutes with the translations on the weight $[U,\Delta_\mathrm{W}] = 0$. The generator of the translations $\Delta_\mathrm{W}$ is canonically conjugated to the position (or energy) of the weight $[H_\mathrm{W}, \Delta_\mathrm{W}] = \i$.

  \item[Classicality of work] Before and after applying the global map $\Gamma_\mathrm{SBW}$ the position of the weight is measured, obtaining outcomes $|x\rangle$ and $|x+W\rangle$ respectively. In general, the work $W$ is a fluctuating random variable.

\end{description}

Condition $[U,\Delta_\mathrm{W}] = 0$  implies that the reduced map on system and bath is a mixture of unitaries (Result 1 in \cite{masanes2014derivation}). Hence, these transformations can never decrease the entropy of system and bath, which guarantees that the weight is not used as a source of free energy.

Let us define the dephasing map as
\begin{equation}
  \Theta_\alpha [\rho_\mathrm{S}] = 
  \int_\mathbb{R}\! dt\, 
  \e^{i\alpha t}\,
  \e^{iH_\mathrm{S} t} \rho_\mathrm{S} \e^{-iH_\mathrm{S} t}
  \ .
\end{equation}
Energy conservation, the classicality of work and the fact that the initial state of the bath commutes with its Hamiltonian imply
\begin{equation}
  \Theta_\alpha \circ \Gamma_\mathrm{S}
  =
  \Gamma_\mathrm{S} \circ \Theta_\alpha 
  \ ,
\end{equation}
where $\Gamma_\mathrm{S} (\rho_\mathrm{S}) = {\rm tr}_\mathrm{BW} \Gamma_\mathrm{SBW} (\rho_\mathrm{S} \otimes \rho_\mathrm{B} \otimes \rho_\mathrm{W})$ is the transformation of the system.
Note the assumption that the initial state of system, bath and weight is uncorrelated.
See \cite{richens2016quantum} for a proof.
Setting $\alpha=0$ we have that, if the initial state of the system commutes with $H_\mathrm{S}$, then so does the final state. And, if the final state of the system commutes with $H_\mathrm{S}$, then so does the initial one. 
In this paper we only consider processes in which one of the two states (and hence both) is diagonal. 
For example, optimal work extraction from an arbitrary initial state is one such process. For processes where the initial and final states involve coherences, our results provide an upper bound to the work. See Appendix A1 for further details.

We write the initial and final states as $\rho_\mathrm{S} = \sum_{s} P(s) |s\rangle\! \langle s|$ and $\rho'_\mathrm{S} = \sum_{s'} P(s') |s'\rangle\! \langle s'|$, respectively. Where $|s\rangle$ and $|s'\rangle$ are the initial and final energy eigenstates. 
Note that we allow initial and final Hamiltonians $H_\mathrm{S}$ and $H'_\mathrm{S}$ to be not necessarily equal.

\subsection{Corrections to the second law}

In this section we analyze the transformation power of the operations defined in the previous section; and study the effect of not having an infinite heat bath. 
First, we present a generalization to the second law~\eqref{eq:landauer} to the case of arbitrary heat bath.
This necessary and sufficient condition is not always useful, because it involves the final state of the bath. However, it precisely articulates the effect of having a finite bath. 
Later, we provide more practical bounds that are independent from the state of the bath.

\medskip
\noindent \textbf{Theorem 1:} \emph{A necessary and sufficient condition for the possibility of a transformation is}
\begin{align}
  \label{T3}
  \langle W\rangle  &\leq
-\Delta F - \frac{1}{\beta}
D[P(E's') \| p_G(E')P(s')]
\ ,
\end{align} 
\emph{where $P (s',E')$ is the joint probability distribution for the final state of the system $s'$ and the energy of the bath $E'$, $P (s')$ is the final state for the system, and $p_G(E')$ is the energy distribution for the Gibbs state.
We use the relative entropy $D[ p(x) \| q(y)]= \sum_{x,y} p(x)\log\! \left[p(x)/q(y) \right]$.}

\medskip
\noindent \emph{Proof.} See Appendix B1.

\medskip\noindent
Note that an equation similar to \eqref{T3} was previously found in  \cite{esposito2010entropy} and \cite{reeb2014improved}. When the heat capacity of the bath is infinite, it is possible to extract average work equal to the free energy difference and a system-bath final state of the form $P (s',E') = P (s') p_G(E')$ is achievable.
When this happens, the system and bath end up uncorrelated, and the bath remains in a thermal state.
The correction of~\eqref{T3} quantifies the distance between the real final state $P (s',E')$, and the ideal one $P (s') p_G (E')$.
Hence, the standard second law only applies when $P (s',E') = P (s') p_G(E')$. 
This idealized situation cannot be achieved with a finite bath (see Appendix B1). 
In finite baths of any size, all optimal work-extraction transformations leave the bath an athermal state correlated with the system.

The bound in Theorem 1 requires detailed knowledge of the final system and bath joint state which is not typically available in a realistic setting. 
We now present a necessary condition which only depends on the state of the system and the heat capacity of the bath $C$. 
This upper bound to the extractable work may not be tight in general, although below (Theorem~3) we explore some cases for which it is. 

\medskip\noindent \textbf{Theorem 2:} \emph{A necessary condition for the possibility of a transformation is}
\begin{equation}\label{eq:modfree2}
\langle W\rangle \leq -\Delta F - \frac{1}{2\beta C}\Delta S^2\ ,
\end{equation}
\emph{where $\Delta S$ is the change in entropy of the system.}\\

\noindent \emph{Proof.} See Appendix B2.
This provides a much tighter bound than the  macroscopic second law~\eqref{eq:landauer}.

\medskip\noindent
In what follows we present a necessary and sufficient condition that only depends on the initial and final states of the system, like~\eqref{eq:modfree2}. This independence from the state of the bath holds up to first order in $1/C$. Hence, it is valid in the regime of large (but not necessarily infinite) heat bath.
In addition, it requires that the initial or final state of the system is the uniform distribution (maximally-mixed state $P(s)=$ const).
For these cases we can refine the work upper bound in~\eqref{eq:modfree2} to a tight upper bound and show that it is achievable through operations that are independent of the bath energy (See Appendix B3).

\medskip\noindent \textbf{Theorem 3:} \emph{When the final state of the system is maximally mixed, the necessary and sufficient condition is}
\begin{equation}
\label{MWE}
\langle W\rangle \leq -\Delta F - \frac{1}{2\beta C}\Delta S^2\ ,
\end{equation}
\emph{up to first order in $1/C$.
When the initial state is maximally, the necessary and sufficient condition is}
\begin{eqnarray}
\label{WC}
\langle W\rangle \leq -\Delta F - \frac{1}{2\beta C}\left(\Delta S^2 +\mathsf{Var}[P(s')] \right)
\ ,
\end{eqnarray}
\emph{up to first order in $1/C$. The \emph{varentropy} is defined as 
}
\begin{equation}
\label{def:varentropy}
\mathsf{Var}[P(s)] = \sum\limits_{s} P(s)\log^2\! P(s) - \left[\sum\limits_{s} P(s) \log P(s)\right]^2 ,
\end{equation}
\emph{and is always positive.}\\

\noindent \emph{Proof.} See Appendix B3.\\

In the rest of this section we discuss the case where the Hamiltonian of the system is trivial $H=0$. This implies that the thermal state at any temperature is the maximally-mixed one.
Hence, two applications of the above two results are: \emph{maximal work extractable} from a state~\eqref{MWE}, and the \emph{work cost of preparing a state} from equilibrium ones~\eqref{WC} (e.g. erasure).

The optimal procedures in terms of work are the ones which saturate the above inequalities.
We see that the minimal work cost for preparing an arbitrary state $P(s')$ minus the maximal work extractable from it is
\begin{equation}
  \frac{1}{2\beta C}
  \left(\Delta S^2 +\mathsf{Var}[P(s')] \right)\ ,
\end{equation}
which is always positive.
Hence, this circular process, despite being optimal, is not reversible.
Therefore, we conclude that
\begin{quote}\em
Thermodynamic reversibility requires an infinite heat bath.
\end{quote}



Indeed this is implied for all processes that change the entropy of the system $\Delta S\neq 0$ by \eqref{eq:modfree2}. The varentropy~\eqref{def:varentropy} arises in second-order Shannon information theory. And it can be interpreted as giving the variance of the ``surprise'' or ``fine-grained entropy'' $-\log P(s)$, whereas the Shannon entropy gives the average of this. 
For thermal states $P(s) \propto \e^{-\beta \mathcal E_s}$, the varentropy is the heat capacity~\eqref{c def} of the system (see section 2.2.2 of \cite{reeb2015tight}). For example, consider the case where we prepare a thermal state of an arbitrary Hamiltonian $H'_S$ from the thermal state of a trivial Hamiltonian. The term $\mathsf{Var}[P(s')]$ gives the heat capacity of the prepared system and the correction includes the ratio of this to the heat capacity of the bath.

\section{Deterministic work}\label{sec:detTO}

For some applications it is preferable that the work extracted from a system does not fluctuate.
This has led some authors to consider a more restrictive definition of work for the quantum regime, namely \emph{deterministic} work, which consists of the raising or lowering of a system from one energy level to another with a very high probability \cite{horodecki2013fundamental,aaberg2013truly,gemmer2015single}. We now explore in which way can deterministic work appear when one does not have an infinite bath.

We shall focus on joint transformations of the system and weight such as
\begin{equation}\label{eq:trans}
\rho \otimes \ket{0}\!\bra{0} \rightarrow \sigma \otimes \ket{W}\!\bra{W},
\end{equation}
where $\ket{0}\! \bra{0}$ and $\ket{W}\! \bra{W}$ are energy eigenstates of the weight with definite energies $0,W$. The figure of merit here is the maximum $W$ that can be achieved through the operations defined in Section~\ref{sec:TOfluc} for each particular case. Due to the lack of work fluctuations, we can now effectively take the weight to be an additional part of our system.

The joint system-bath-weight has a well-defined total energy $E_{\rm tot}$ with probability distribution $P(E_ {\rm tot})$. 
Because the joint transformation conserves total energy, a transition is possible in general if it is also possible separately in each subspace of fixed $E_{\rm tot}$ (as the unitary acts separately and independently on each subspace). This is different to the previous secion, where we were able to consider optimal averages over all total energies. In Appendix \ref{app:det} we describe a criteria for transitions to be possible for each such subspace, which reduces to thermomajorization (the full criteria when the bath is infinite) at the average total energy $\langle E_{\rm tot} \rangle$. Which, without loss of generality we set it to zero $\langle E_{\rm tot} \rangle =0$.
Because of the dependence on the total energy, and unlike in the infinite-bath case, we find that conclusive answers as to which transitions are possible cannot be given in terms of thermo-majorization. 

However, we are able to obtain a nontrivial answer for particular processes if we allow for a small probability of error that comes from ignoring the tails of the distribution $P(E_{\rm tot})$, such that we only consider a finite energy range around the average. Ignoring events with small probability is a standard practice single-shot information theory~\cite{renner2008security,tomamichel2015quantum}, although conceptually it is an additional complication with respect to the infinite-bath regime.

In Appendix~\ref{app:dist} we show that for an energy range $E_{\rm tot} \in [ -E^* , E^*]$, the probability of failure $\epsilon$ is approximately given by
\begin{equation}
\label{eq:eps1}
  \epsilon \simeq 
  \frac{2^{3/2}}
  {\sqrt{\pi \gamma} E^*} 
  e^{-\frac{\gamma}{2} {E^*}^2},
\end{equation}
where to derive this we assume that the energy fluctuations of the bath are much larger than those of the system. That is $\gamma^{-1/2} \gg \|H_S \|_\infty$, where the operator norm gives the largest eigenvalue in absolute value.


Even with this restriction we cannot give conclusive answers to general transitions, but it is possible if one takes either the initial or the final state to be thermal. In this case, the criteria simplifies, as we only have to look at the extremal points of the energy distribution (this is shown in Appendix \ref{app:detwork}). Hence, given that we allow for a probability of failure, one can compute the maximum work that one can extract in the transition that takes a state $\rho = \sum_s P(s) \ketbra{s}{s}$ to the thermal state, as well as the minimum work needed in the opposite transition, when creating a state $\rho = \sum_s P(s) \ketbra{s}{s}$ from a thermal state. 
We denote the Hamiltonian of the system by $H = \sum_s \mathcal E_s \ketbra{s}{s}$, the thermal state by $\tau_\beta = \frac 1 {Z_\beta} e^{-\beta H_{\rm S}}$, and the partition function by $Z_{\beta}= \sum_s e^{-\beta \mathcal E_s}$.
We hence find the following two results (the details can be found in Appendix \ref{app:detwork}) :

\medskip
\noindent \textbf{Theorem 4:} \emph{(Work extraction)}\label{result1}
\emph{
The maximal deterministic work that one can extracted from sate $P(s)$}
\emph{is, up to error $\epsilon$, given by $W_{\text{ext}}^\epsilon= F^{\beta_-}_{\min}(\rho)$, where $\beta_- = \beta-\gamma E^*$ and}
\begin{equation}  
  F_{\min}^{\beta} (\rho) =
  \frac 1 {\beta}
  \log Z_{\beta} 
  -  \frac 1 {\beta}
  \log\!\left(
  \sum_s e^{-{\beta} \mathcal E_{s}} P(s)^0 \right)
 ,
\end{equation}
\emph{and the relation between $\epsilon$ and $E^*$ is given by Eq.~\eqref{eq:eps1}}.

\bigskip\noindent\emph{Proof.}
See Appendix C2. Note that we make us of the algebraic identities $x^0 = 1$ if $x>0$ and $x^0 = 0$ if $x=0$.

\bigskip\noindent \textbf{Theorem 5:} \emph{(Work of formation)
In the transition }
\begin{equation}\label{eq:trans2}
\tau_\beta \otimes \ket{W}\bra{W} \rightarrow \rho \otimes \ket{0}\bra{0}\ ,
\end{equation}
\emph{the minimum possible value of $W$ is given by} 

\begin{equation}
  W_{\text{for}}^\epsilon =\frac{\beta}{\beta_+}
  F_{\max}^\beta(\rho)- 
  \frac{1}{\beta_+}
  \log\! \left(
  \frac{Z_{\beta}}{Z_{\beta_+}}
  \sum_{s}  P(s) e^{(\beta -\beta_+) \mathcal E_{s}}
  \right)
\end{equation}
\emph{where $\beta_+=\beta+\gamma E^*$ and}
\begin{equation}
  F_{\max}^{\beta}(\rho) =
  \frac 1 \beta
  \log\max_s P(s) e^{{\beta} \mathcal E_{s}} \ .
\end{equation}

\bigskip\noindent\emph{Proof.}
See Appendix C2.

We see that both quantities converge to the results of \cite{horodecki2013fundamental} in the infinite limit where $\gamma \rightarrow 0$, where one has $F_{\min}^{\beta}(\rho)$ and $F_{\max}^{\beta}(\rho)$ as the extractable and the work of formation. 

The fact that in this case we need to allow for a probability of error is a consequence of the 3rd law \cite{masanes2014derivation,alhambra2016second}, as in general only with an infinite bath (with degrees of freedom that require infinite time to be reached) can perfectly deterministic work be extracted or expended. We note, however, the $\epsilon$ that appears here is not the same as that of the \emph{smoothed} version of the $F_{\min}^{\beta}(\rho)$ and $F_{\max}^{\beta}(\rho)$ free energies \cite{horodecki2013fundamental,renner2008security}. There, the small error probability does not come from cutting off the distribution of energies of the bath, but from optimizing over an $\epsilon'$-sized ball in the space of states, such as some $\rho_{\epsilon'}$ for which $|| \rho_{\epsilon'}-\rho ||_1 \le \epsilon'$. Hence, in our expressions, we can also implement this further smoothing too, such that they depend on both $\epsilon$ from Eq.\eqref{eq:eps} and $\epsilon'$ from the smoothing of the state. The optimal values of the work will then be 
\begin{equation}
W_{\epsilon,\epsilon'}= \frac{1}{\beta_-} \sup_{\rho_{\epsilon'}}  F_{\text{min}}^{\beta_-}(\rho_{\epsilon'}) 
\end{equation}
for the extractable one and
\begin{align*}
W_{\epsilon,\epsilon'}\!&=\!\frac{1}{\beta_+}\log{\frac{Z_{\beta_+}}{Z_\beta}}\\
&+\frac{1}{\beta_+}\inf_{\rho_{\epsilon'}} \Big[ F_{\text{max}}^{\beta}(\rho_{\epsilon'})+\log{\frac{1}{\sum_{s} P_{\epsilon'} (s) e^{(\beta-\beta_+) \mathcal E_{s}}}}\Big] \numberthis 
\end{align*}
for the work of formation. We here define $P_{\epsilon'}$ as the probability spectrum of the state $\rho_{\epsilon'}$, and recall that $\beta_{\pm}=\beta \pm \gamma E^*$.

In Fig. \ref{fig:EpsW} of Appendix \ref{app:detwork} we show an example of the tradeoff between these works and the probability of failure $\epsilon$ allowed.

\section{Conclusion}

In this paper we have derived the finite-bath corrections to the work that can be extracted or expended in a thermodynamical transition, in the cases where the work is taken as a fluctuating quantity and when it is taken as a definite value. Our approach is general in the sense that we do not need to consider the particular microscopic structure of the heat bath (e.g. whether it is made of fermions, bosons,...). The only quantities that play a role are its temperature and its heat capacity. When the heat capacity diverges, i.e. the bath becomes infinite, one recovers all the standard results, such as the result in \cite{skrzypczyk2014work,aaberg2013truly} for fluctuating and deterministic work.

Previous work on finite-size limitations includes \cite{reeb2014improved,jakvsic2014note,pekola2016finite,woods2015maximum}. For example in \cite{reeb2014improved} tight corrections to the Landauer bound are found in terms of the dimension of the bath, recovering the Landauer limit $\langle W \rangle \leq \Delta S$ when the heat bath is infinite dimensional. This work concerns bounding the size of the bath, by which we mean its volume, and for this the dimension is not a relevant quantity. On the other hand, the heat capacity of the bath is proportional to its volume. To illustrate this difference, there are situations where the bath is infinite dimensional but with a finite volume and heat capacity (for example a box of air with finite volume, a bosonic bath, and so on...). In these cases bounding dimension results in trivial corrections to the free energy whereas the heat capacity provides non-trivial corrections.  Furthermore, the fact that the bath has infinite dimensions is also a necessary condition for the appearance of an infinite recurrence time, which is needed for the emergence of irreversibility \cite{bocchieri1957quantum}. Working in the regime of finite but large environment, we find corrections to the second law that are universal, in that they apply to all concievable environments, regardelss of their constituents, Hamiltonian and Hilbert space dimension. 

One avenue to explore is how our results limit the efficiency of heat engines in finite time. We have observed that the effect of having access to a finite bath has a marked effect on the minimal achievable dissipation in some protocols, suggesting that there could be rich and unexplored finite size effects in small scale thermal engines. A correction to the free energy should also correspond to a correction of the Carnot efficiency of ideal cyclic process. Previous work such as \cite{woods2015maximum,skrzypczyk2011smallest,uzdin2015equivalence} has looked at this question for particular models. An open question is to determine the work-optimal processes for arbitrary state transformations with a finite bath. Whereas for an infinite bath the work is a function of the state, and as a result the work is ``path independent'', for finite baths our results suggest the existence of unique optimal processes that warrant further study.

In our setting we do not consider the possibility of coherence. Given the energy conservation restriction imposed it is known that this means no work can be extracted from states with coherence. It would be interesting to see how the presence of coherence provides additional constraints to finite-bath work extraction. Some results for deterministic work can be found in \cite{korzekwa2016extraction}, and the impact of coherence in heat engines is examined in \cite{uzdin2015equivalence,mitchison2015coherence,gardas2015thermodynamic}.

Finally, we have seen that in the case of Landauer erasure and state formation, second order information measures are required to compute tight upper bounds on work. It is often stated that thermodynamics has deep roots in information theory, and these results suggest that in order to move away from asymptotic approximations in thermodynamics we must use information measures that take these non-asymptotic effects into account. In order to obtain more directly applicable corrections, one may have to look at particular protocols and particular models of the bath. As one considers increasingly more terms in the expansion of the density of states Eq. \ref{eq:probu} in order to get a more accurate result, an increasingly more detailed knowledge of the microscopic features of the bath is required.

\bigskip\noindent
{\bf Acknowledgements.}
LM and JR are funded by EPSRC. AMA acknowledges support from the FQXi. The authors would like to thank David Reeb for useful discussions.

\bibliographystyle{apsrev4-1}
\bibliography{References}

\clearpage
\widetext
\appendix

\section{Derivation of the modified free energy: Preliminaries}

\subsection{Thermal operations with fluctuating work}\label{TOF}
We first characterize the type of thermodynamic transformation that we consider, which we refer to as \emph{thermal operations with fluctuating work}. We make use of a widely applied set-up for defining the work of a thermodynamical transformation \cite{skrzypczyk2014work, masanes2014derivation, gemmer2015single,alhambra2016second}. 
Our setting consists of a system with Hamiltonian $H_{\rm S}$, the bath with Hamiltonian $H_{\rm B}$ initially in the thermal state, and an ideal weight with Hamiltonian $H_\mathrm{W}=  \int_\mathbb{R} dx\, x |x\rangle \! \langle x|$, where the orthonormal basis $\{|x\rangle, \forall\, x\in \mathbb R\}$ represents the position of the weight. Any joint transformation of system, bath and weight is represented by a Completely Positive Trace Preserving (CPTP) map $\Gamma_\mathrm{SBW}$ satisfying the following conditions:
\begin{description}
  \item[Microscopic reversibility (Second Law)] It has an (CPTP) inverse $\Gamma_\mathrm{SBW}^{-1}$, which implies unitarity $\Gamma_\mathrm{SBW} (\rho_\mathrm{SBW}) = U\rho_\mathrm{SBW} U^\dagger$.
    
  \item[Energy conservation (First Law)] 
  $[U,H_\mathrm{S} +H_\mathrm{B} +H_\mathrm{W}] =0$. 

  \item[Independence from the ``position" of the weight]
  The unitary commutes with the translations on the weight $[U,\Delta_\mathrm{W}] = 0$. The generator of the translations $\Delta_\mathrm{W}$ is canonically conjugated to the position (or energy) of the weight $[H_\mathrm{W}, \Delta_\mathrm{W}] = \i$.

  \item[Classicality of work] Before and after applying the global map $\Gamma_\mathrm{SBW}$ the position of the weight is measured, obtaining outcomes $|x\rangle$ and $|x+W\rangle$ respectively. In general, the work $W$ is a fluctuating random variable.

\end{description}
Let us define the dephasing map as
\begin{equation}
  \Theta_\alpha [\rho_\mathrm{S}] = 
  \int_\mathbb{R}\! dt\, 
  \e^{i\alpha t}\,
  \e^{iH_\mathrm{S} t} \rho_\mathrm{S} \e^{-iH_\mathrm{S} t}
  \ .
\end{equation}
Energy conservation, the classicality of work and the fact that the initial state of the bath commutes with its Hamiltonian imply
\begin{equation}
  \Theta_\alpha \circ \Gamma_\mathrm{S}
  =
  \Gamma_\mathrm{S} \circ \Theta_\alpha 
  \ ,
\end{equation}
where $\Gamma_\mathrm{S}$ is the transformation of the system.
See \cite{richens2016quantum} for a proof.
Setting $\alpha=0$ we have that, if the initial state of the system commutes with $H_\mathrm{S}$, then so does the final state. And, if the final state of the system commutes with $H_\mathrm{S}$, then so does the initial one. 
In this paper we only consider processes in which one of the two states (and hence both) is diagonal. For example, optimal work extraction is one such process. For processes where the initial and final states involve coherences, our results provide an upper bound to the work.

%

Let us define the stochastic matrix
\begin{equation}\label{t}
  t(s',E'|s,E) =     
  {\rm tr}\! \left[ \left(|s'\rangle\! \langle s'|
  \otimes Q_{E'} \otimes \mathbb{I} \right) U 
  \left(|s\rangle\! \langle s| \otimes 
  \frac{Q_E} {\Omega(E)} \otimes \rho_W
  \right) U^\dagger \right] \ ,
\end{equation}
where $Q_E$ is the projector onto the eigenspace of $H_\mathrm{B}$ with energy $E$, and $\Omega(E) = {\rm tr} Q_E$ is the density of states.
This matrix only contains partial information about $U$, but this is enough to derive relevant constraints for any transformation of the type described above.

The energy eigenvectors of the system are labeled by $s$ and the corresponding eigenvalues are $\epsilon_s$. 
Let us derive some properties for the map~\eqref{t}. 
The average work extracted by the map is given by 
\begin{equation}\label{eq:avgwork}
\langle W \rangle= \sum\limits_{EE'ss'}p(s)p_b(E)t(E'\,s'\, | \, E\, s)[(E-E') +(\epsilon_s - \epsilon_{s'})]
\ ,
\end{equation}
where $p(s)$ is the given initial state of the system, and $p_b(E) = \frac 1 Z \Omega(E) e^{-\beta E}$ is the probability of finding the bath in the energy subspace $E$.
It is easy to check that the map~\eqref{t} satisfies
\begin{empheq}[box=\widefbox]{align}
\sum\limits_{E\, E'\, s}p(s)p_b(E)t(E'\,s'\, | \, E\, s) &= q(s') \quad \forall \, s' \label{1} \\
\sum\limits_{E' \, s'} t(E'\,s'\, | \, E\, s) &= 1 \quad \forall \, E, \, s\label{2} \\
  t(E'\,s'\, | \, E\, s) &\geq 0  \quad \forall E, E', s, s' \label{3}\\
\sum\limits_{E\, s} t(E'\,s'\, | \, E\, s)\frac{\Omega(E)}{\Omega(E')} &= 1 \quad \forall \quad s', \, E' \label{4}
\end{empheq}
The first condition is that the reduced map on the system achieves the desired state transformation $\Gamma_S(\rho_S)=\rho'_S$, with the second and third conditions ensuring that $t$ is a stochastic matrix. 
The fourth constraint follows from microscopic reversibility. This can be interpreted as the map $t$ being one-to-one on set of joint system-bath states, where $\sum_{E\, s} t(E'\,s'\, | \, E\, s)\Omega(E)$ is the number of states mapped to joint energy subspace $(E',s')$, which has degeneracy $\Omega (E')$. For similar uses of this set-up see \cite{masanes2014derivation,richens2016quantum}.
Next we show that the equalities in the reversibility constraint can be replaced by inequalities.


\subsection{Thermal operations with non-constant Hamiltonian}
\label{ss:toH}

Thermal operations are general enough to include the case where the initial Hamiltonian of the system $H_\mathrm{S}$ is different than the final one $H'_\mathrm{S}$. This is done by including an additional qubit $X$ which plays the role of a switch (as in \cite{horodecki2013fundamental,aberg2016fully}). Now the total Hamiltonian is
\begin{equation}
  H = H_\mathrm{S}\otimes |0\rangle _\mathrm{X}
  \langle 0| + H'_\mathrm{S} \otimes |1\rangle_\mathrm{X}
  \langle 1| +H_\mathrm{B} +H_\mathrm{W}
\ ,
\end{equation}
and energy conservation reads 
$[V,H] =0$, where $V$ is the global unitary when we include the switch.
We impose that the initial state of switch is $|0\rangle_\mathrm{X}$ and the global unitary $V$ performs the switching
\begin{equation}
  \label{SS}
V \left( 
\rho_\mathrm{SBW} \otimes |0\rangle_\mathrm{X} \langle 0|\right)  V^\dagger = \rho'_\mathrm{SBW} \otimes |1\rangle_\mathrm{X} \langle 1| 
\ ,
\end{equation}
for any $\rho_\mathrm{SBW}$.
This implies
\begin{equation}
  V=  U
  \otimes |1\rangle_\mathrm{X} \langle 0| +
  \tilde U
  \otimes |0\rangle_\mathrm{X} \langle 1|
  \ ,
\end{equation}
where $U$ and $\tilde U$ are unitaries on system, bath and weight. Condition $[V,H] =0$ implies
\begin{equation}
  \label{EC}
  U (H_\mathrm{S}+H_\mathrm{B}+H_\mathrm{W}) =
  (H'_\mathrm{S}+H_\mathrm{B}+H_\mathrm{W}) U
  \ .
\end{equation}
Therefore, the reduced map on system, bath and weight can be written as
\begin{equation}
  \label{rmm}
  \Gamma_\mathrm{SBW} (\rho_\mathrm{SBW}) =
  U \rho_\mathrm{SBW} 
  U^\dagger\ ,
\end{equation}
where the unitary $U$ does not necessarily commute with $H_\mathrm{S}+H_\mathrm{B}+H_\mathrm{W}$ nor $H'_\mathrm{S}+H_\mathrm{B}+H_\mathrm{W}$ but satisfies~\eqref{EC}.

\section{Upper bounds on work}\label{deriv}

In this appendix we prove an upper bound for the work extracted by a thermal operation. The only assumption on the bath is that the density of states $\Omega (E)$ is convex. Hence this applies to a large class of baths including finite and infinite ones.

\subsection{General upper bound}

The following theorem establishes an upper bound for the work in terms of how similar is the system-bath's final state  $P(E's')$ to the product state $p_G(E')P(s')$ where $p_G(E')$ is the Gibbs state and $P(s')$ is the marginal of $P(E's')$. Below we show that this can only happen when the baht is infinite.



\medskip
\noindent 
\textbf{Theorem 1:}\label{upper bound lemma} \emph{ The work extracted by the thermal operation $t(E's'|Es)$ is upper bounded by }
\begin{align}
  \label{T3}
W &\leq
-\Delta F 
- \frac{1}{\beta}
D[P(E's') \parallel p_G(E')P(s')]\ ,
\end{align} 
\noindent \emph{Where $D[x \parallel y]= \sum_{x,y} p(x)\log \left(p(x)/q(y) \right)$ is the relative entropy, $p_G(E')$ is the probability distribution of the bath energy $E'$, $P(s')$ is the final probability distribution of the dephased system, and $P(E's')$ is the joint final distribution of system and bath. }

\begin{proof}
For each thermal operation $t(E's'|Es)$, we define the following two functions
\begin{align}
R(E's'|s) &= \sum\limits_E p_G(E) t(E's'|Es) \label{mapping}\ ,
\\
Q(E' s' s) &= \frac{\sum\limits_E p_G(E) t(E's'|Es) (E-E')}{R(E's'|s)}
\ . \label{Q def}
\end{align}
The work extracted by $t(E's'|Es)$ can be expressed in terms of $R,Q$ in the following way
\begin{align}
  W &= \sum\limits_{EE'ss'}
  P(s)p_G(E) R(E's'|Es) (E-E' +\mathcal E_s - \mathcal E_{s'})
  \nonumber \\ \label{work} 
  &= \sum\limits_{E'ss'} P(s) 
  R(E's'|s) Q(E's's) +\Delta U
\end{align}
($\Delta U$ taken to be minus the change in internal energy of the system) where we have used \eqref{1},  \eqref{2} and
\begin{equation}
\sum\limits_{EE'ss'}P(s)p_G(E)t(E's'|Es)(\mathcal E_s-\mathcal E_{s'}) = \sum\limits_s P(s) \epsilon_s - \sum\limits_{s'} P(s') \epsilon'_{s'}=\Delta U
\end{equation}
The condition for microscopic reversibility \eqref{4} implies the following conditions on $R$ and $Q$
\begin{align*}
1 &= \sum\limits_{Es}t(E's'|Es)\frac{\Omega (E)}{\Omega(E')}\\
&= \sum\limits_{Es}t(E's'|Es)\frac{p_G (E)}{p_G(E')}e^{\beta (E-E')}\\
&= \sum\limits_{Es}\frac{R(E's'|s)}{R(E's'|s)}t(E's'|Es)\frac{p_G (E)}{p_G(E')}e^{\beta (E-E')}\\
&\geq \sum\limits_s \frac {R(E's'|s) } {p_G(E')}
  \exp\!\left[ {\beta \frac{\sum_E p_G(E) t(E's'|Es) (E-E')}{R(E's'|s)}} \right]\\
  \numberthis \label{condi RQ}
&= \sum\limits_s R(E's'|s) \frac{e^{\beta Q(E's's)}}{p_G(E')}
\ .
\end{align*}
Note that the only information from $t(E's'|Es)$ that appears in the bound \eqref{T3} is $P(E' s')$, which is fully contained in $R$. Hence, in order to obtain this bound, we optimize over all possible $Q$s subject to constraint \eqref{condi RQ}, and keep $R$ fixed. 
The value of the work for the optimal $Q$ is found by maximizing the Lagrangian 
\begin{equation}
\mathcal L =  \sum\limits_{E'ss'}P(s) R(E's'|s) Q(E' s' s) + \sum\limits_{E's'}\lambda_{E's'}\left(p_G(E') - \sum\limits_s R(E's'|s)e^{\beta Q(E's's)} \right)\ .
\end{equation}
Taking the derivative over $Q(E's's)$  and equating to zero we obtain
\begin{equation}
Q(E's's) = \frac{1}{\beta}\log \left(\frac{P(s)}{\lambda_{E's'}\beta} \right)
\end{equation}
substituting this back into the Lagrangian and taking the derivative w.r.t the Lagrange multiplier $\lambda_{E's'}$ and equating to zero gives 
\begin{equation}
\frac{1}{\lambda_{E's'}\beta} 
= 
\frac{p_G(E')}{\sum\limits_{s}P(s) R(E's'|s)}
= 
\frac{p_G(E')}{P(E's')}\ ,
\end{equation}
where for the last equality we used the definition of $R$ given in  \eqref{mapping}.
This gives the optimal $Q$
\begin{equation}
Q(E's's) = \frac{1}{\beta}\log \left(\frac{P(s)p_G(E')}{P(E's')} \right) \label{Q}\ .
\end{equation}
And thus the optimal work 
\begin{equation}
W \leq \Delta U + \frac{1}{\beta}\sum\limits_{E'ss'}P(s) R(E's'|s) \log \left(\frac{P(s)p_G(E')}{P(E's')} \right) \ ,
\end{equation}
which can also be written as
\begin{eqnarray}
W&\leq& \Delta U -\frac{1}{\beta}\Delta H - \frac{1}{\beta}\sum\limits_{E'ss'}P(s) R(E's'|s)\log \left(\frac{P(s)}{P(s')}\right)\\ 
&+& \frac{1}{\beta}\sum\limits_{E'ss'}P(s)  R(E's'|s) \log \left(\frac{P(s)p_G(E')}{P(E's')} \right) \, , 
\end{eqnarray}
where $\Delta H=H(\rho) - H(\rho')$, which simplifies to

\begin{equation}
W\leq -\Delta F - \frac{1}{\beta}
\sum\limits_{E's'}P(E's')\log\left( \frac{P(E's')}{p_G(E')P(s')}\right) \label{general expression}
\end{equation}
where we have used \eqref{1} and \eqref{2}.  

\end{proof}

This shows that the full free energy can only be extracted as work when: (i) the state of the bath remains thermal $P(E') = p_G(E')$, and (ii) system and bath end up uncorrelated $P(E's')=p_G(E')P(s')$. 
In the following section we prove that these two ideal conditions can only be achieved when the bath is infinite. 
We note that the bound \eqref{T3} requires knowledge of the final joint state of system and bath. Te following lemma gives a different bound that only depends on the initial and final states of the system and $\gamma$, and it is tight.

\subsection{Tight upper bound for finite baths}

First we address the case of infinite bath ($\gamma =0$). 

\begin{lemma}
When $\gamma=0$ all thermal operations  saturating the inequality $\langle W \rangle \geq -\Delta F$ are of the form \[ t(E's'|Es)=f(E's's)\delta (E-E'-f_{s'}+f_s)\label{infinite optimal} \] where $f(E's's)$ obeys
\begin{eqnarray}
\sum\limits_{s'}f(E-f_{s'}+f_s, s', s) &=& 1 \quad \forall \  s \label{new 1} \\
\sum\limits_{s}P(s) f(E's's) &=& P(s') \quad \forall \ s' 
\label{new 2} \\
0 \leq f(E's's) &\leq& 1 \quad \forall \ E's's \label{new 3}
\end{eqnarray}
where $f_s = -\beta^{-1}\log P(s)$ is the fine-grained entropy.
\end{lemma}

\begin{proof}
In order to achieve $W=-\Delta F$, a thermal operation $t(E's'|Es)$ must saturate the upper bound derived in Theorem 1, and the right hand side of the bound must be equal to $-\Delta F$. First note that in Theorem 1 the optimal $Q(E's's)$ are given uniquely by 
\begin{equation}
Q(E's's) = \frac{1}{\beta}\log \left(\frac{P(s) p_G(E')}{p(E's')} \right)\ .
\end{equation}
By Theorem \ref{upper bound lemma}, as the relative entropy is a distance measure for probability distributions the upper bound is $-\Delta F$ if and only if $P(E's')=p_G(E')P(s')$ and therefore, for all optimal thermal operations
\begin{equation}
Q(E's's) = \frac{1}{\beta}\log \left(\frac{P(s) }{P(s')} \right)=f_{s}-f_{s'}\ .
\end{equation}
Finally, note that the thermal operation only construct an optimal $Q$ if the reversibility constraint \eqref{condi RQ} is saturated. This requires that 
\begin{equation}
\sum\limits_{Es}\frac{R(E's'|s)}{p_G(E')}t(E's'|Es)\frac{p_G (E)}{R(E's'|s)}e^{\beta (E-E')}\\
= \sum\limits_s \frac {R(E's'|s) } {p_G(E')}
  \exp\!\left[ {\beta \frac{\sum_E p_G(E) t(E's'|Es) (E-E')}{R(E's'|s)}} \right]
\end{equation}
Therefore the sum over $E$ in the exponent must have only a single term by convexity, and therefore $t(E's'|Es)$ be a delta function killing the sum over $E$. This, combined with the definition of $Q(E's's)$ in \eqref{Q def} implies that $t(E's'|Es)$ is of the form
\begin{equation}
t(E's'|Es) = f(E's's)\delta (E-E'-f_{s'}+f_s) 
\end{equation}
where $f(E's's)$ is a function that, by substituting \eqref{infinite optimal} into \eqref{1},\eqref{2},\eqref{3} obeys
\begin{eqnarray}
\sum\limits_{s'}f(E-f_{s'}+f_s, s', s) &=& 1 \quad \forall \  s \label{stoc f 1}\\
\sum\limits_{s}P(s) f(E's's) 
p_G (E'+f_{s'}-f_s)
&=& P(s') p_G (E')\quad \forall \ s', E'
\label{stoc f 2}
\\
0 \leq f(E's's) &\leq& 1 \quad \forall \ E's's 
\end{eqnarray}
In the limit $\gamma\to 0$ the distribution $p_G (E)$ tends to a constant. Hence, equality~\eqref{stoc f 2} becomes~\eqref{new 2}.
\end{proof}
As we now know the form of any optimal thermal operation to zero order in $\gamma$, we can therefore express any optimal thermal operation to order $\gamma$ as a perturbation 
\[ 
t(E's'|Es) = f(E's's)\delta (E-E'-f_{s'}-f_s) +\gamma t_1(E's'|Es) 
\]
Substituting this into (\ref{1}-\ref{3}) and using (\ref{new 1}-\ref{new 3}) we get the following constraints for the general correction to the map

\begin{eqnarray}
\sum\limits_{E's'}t_1(E's'|s) &=& 0 \quad \forall \  s \label{1 pert} \\
\sum\limits_{s}P(s) t_1(E's'|s) &=& 0 \quad \forall \ s'  \label{2 pert} 
\end{eqnarray}
\begin{eqnarray}
\sum\limits_{E's'}t_1(E's'|s) &=& 0 \quad \forall \  s \label{1 pert} \\
\sum\limits_{E',s,E}P(s) t_1(E's'|s) &=& 0 \quad \forall \ s'  \label{2 pert} 
\end{eqnarray}

Using this we now find an expression for the optimal work in terms of $f(E's's)$ only.

\begin{lemma}
For any given transformation $P(s) \to P(s')$ there is a map $t(E's'|Es)$ which extracts work
\[W = -\Delta F - \frac{\gamma^2}{2\beta}\sum\limits_{E's'}P(s')p_G(E') E'^2 \left(  \frac{\sum\limits_s P(s) f(E's's)(f_{s'}-f_s)}{P(s')} \right)^2+\mathcal O(\gamma^2)\ ,
\] 
up to first order in $\gamma$. (Recall that $p_G (E') E'^2$ is of order $1/\gamma)$. \label{optimal work equality}
\end{lemma}

\begin{proof}
The upper bound to the work in Theorem 1 can be written as
$W\leq -\Delta F -\beta^{-1}\sum\limits_{E's'}P(E',s')\log \left(\frac{P(E',s')}{p_G(E')P(s')} \right)$. To first order in $\gamma$, optimal thermal operations are of the form  
\[t(E's'|Es) = f(E's's)\delta (E-E'-f_{s'}+f_s)+\gamma t_1(E's'|Es) \label{expanded t} \ .
\]

Let us calculate the joint final probability distribution $P(E's')$ to first order in $\gamma$. First,
\begin{align*}
 P(E's')&=\sum_{Es}P(s)p_G(E)t(E's'|Es)\\
  &= \sum\limits_s P(s) p_G(E'+f_{s'}-f_s)f(E's's) + \gamma 
\sum\limits_{s}P(s) p_G(E') t_1(E's'|Es)
\ .
\end{align*}
Second, we expand $p_G(E'+f_{s'}-f_s)$ to first order in $\gamma$
\begin{eqnarray}
p_G(E'+f_{s'}-f_s) = p_G(E')(1-\gamma E' (f_{s'}-f_s) + \gamma/2(E'^2\gamma -1)(f_{s'}-f_s)^2)+\mathcal O (\gamma^{3/2})\label{p expansion}
\end{eqnarray}
Note we have taken $E'\sim\gamma^{-1/2}$ as $\sum_{E'} p_b(E') E'^2 = \gamma^{-1}$. Therefore, to first order in $\gamma$, we recover the joint final probability distribution 
\begin{eqnarray}
P(E's') &=& \sum\limits_s P(s) p_G(E') f(E's's) \left(1-\gamma E' (f_{s'}-f_s) + \gamma/2(E'^2\gamma -1)(f_{s'}-f_s)^2 \right)\nonumber \\
&+&\gamma \sum\limits_{s}P(s)p_G(E')t_1(E's'|s)+\mathcal O (\gamma^{3/2})
\end{eqnarray}
and therefore 
\begin{eqnarray}
\frac{P(E's')}{P(s')p_G(E')}-1 &=&-\underbrace{\frac{\gamma E'}{P(s')}\sum\limits_s P(s)f(E's's)(f_{s'}-f_s)}_{[A]}+\underbrace{\frac{\gamma}{2P(s')}(E'^2\gamma -1)\sum\limits_s P(s)f(E's's)(f_{s'}-f_s)^2}_{[B]}\\
&+& \underbrace{\frac{\gamma}{P(s')p_G(E')} \sum\limits_{sE}P(s)p_G(E)t_1(E's'|Es)}_{[C]} \label{p(e,s)}
\end{eqnarray}
where we have used \eqref{stoc f 1}.

Define
\begin{equation}
x(E's') = \frac{P(E's')}{p_G(E')P(s')} -1 \ ,
\end{equation}
and note that $x(E's')\sim\mathcal O (\gamma)$. Using $\log (1+x(E's'))= x-1/2 x^2 +\mathcal O (x^3)$ as $x\ll 1$ we can expand equation \eqref{general expression} to first order in $x(E's')$, giving a correction to the free energy of 
\begin{eqnarray}
- \frac{1}{\beta}\sum\limits_{E's'}P(E's')\log\left( \frac{P(E's')}{p_G(E')P(s')}\right) &=& - \frac{1}{\beta}\sum\limits_{E's'}P(E's')\left(\frac{P(E',s')}{p_G(E')P(s')} -1  \right)\nonumber \\
&+&\frac{1}{\beta}\sum\limits_{E's'}P(E's')\left(\frac{P(E's')}{p_G(E')P(s')} -1  \right)^2 + \mathcal O (\gamma^2) \label{work correction}
\end{eqnarray}
We now substitute here to first order in $\gamma$. First, we show that the contribution from term $[C]$ in \eqref{p(e,s)} is zero, and hence we can choose an optimal thermal operation whereby $E$ is specified by $E',s',s$ and the upper bound to work is saturated. As term $[C]$ is $\mathcal O (\gamma)$, its contribution to the work appears in the $\mathcal O (x)$ term in \eqref{work correction} only
\begin{eqnarray}
-\frac{1}{\beta}\sum\limits_{E's'}P(E's')\left(\frac{P(E',s')}{p_G(E')P(s')} -1  \right)
\end{eqnarray}
To zero'th order in $\gamma$, $P(E's') = \sum_s P(s) f(E's's)p_b(E')= P(s')p_b(E')$. Substituting this and $[C]$ into the above expression gives 
\begin{eqnarray}
-\frac{1}{\beta}\sum\limits_{E's'}p_G(E')P(s')\frac{\gamma}{P(s')p_G(E')} \sum\limits_{sE}P(s)P_G(E)t_1(E's'|Es) =-\frac{\gamma}{\beta}\sum\limits_{E's'Es}P(s)P_G(E)t_1(E's'|Es)
\end{eqnarray}
Applying constraints \eqref{stoc f 1} and \eqref{stoc f 2} sets this term to zero. Substituting $[A]$ and $[B]$ into \eqref{work correction} and working to $\mathcal O (\gamma )$ gives
\begin{align*}
W &= -\Delta F  - \frac{1}{\beta}\sum\limits_{E's'}P(E's')\left(\frac{P(E',s')}{p_G(E')P(s')} -1  \right)+ \frac{1}{\beta}\sum\limits_{E's'}P(E's')\left(\frac{P(E's')}{p_G(E')P(s')} -1  \right)^2 + \mathcal O (\gamma^2)\\
&= -\Delta F - \frac{\gamma^2}{\beta}\sum\limits_{E's'}P(s')p_G(E')E'^2\left(\frac{\sum\limits_s P(s)f(E's's)(f_{s'}-f_s)}{P(s')} \right)^2-\frac{\gamma}{2\beta}\sum\limits_{E's'}p_G(E')P(s')(E'^2\gamma-1)\sum\limits_s \frac{P(s)f(E's's)(f_{s'}-f_s)^2}{P(s')}
\end{align*}
where we have used \eqref{stoc f 1}, \eqref{stoc f 2}, \eqref{p expansion} and $E'\sim \mathcal O (\gamma^{-1/2})$ as $\langle E'^2\rangle = \gamma^{-1}$. 

We now show that the second term is equal to zero to first order in $\gamma$
\begin{equation}
\frac{\gamma}{2\beta}\sum\limits_{E's'}p_G(E')P(s')(E'^2\gamma-1)\sum\limits_s \frac{P(s)f(E's's)(f_{s'}-f_s)^2}{P(s')} = 0
\end{equation}
First note that $\sum_{EE'ss'}P(s)p_G(E')t(E's'|Es)=1$ and therefore 
\begin{equation}
1 = \sum\limits_{E'ss'}P(s)p_G(E'-f_{s'}+f_{s})f(E's's) + \gamma \underbrace{\sum\limits_{EE'ss'}P(s)p_G(E)t_1(E's'|Es)}_{ \text{$=0$ by \eqref{1 pert} and \eqref{2 pert}}}
\end{equation}
Expanding $p_G(E'-f_{s'}+f_s)$ gives 
\begin{align*}
1 &= \sum\limits_{E's's}P(s)p_G(E')(1-E'\gamma (f_{s'}-f_s)+\frac{\gamma}{2}(E'^2\gamma-1)(f_{s'}-f_s)^2) f(E's's) +\mathcal O(\gamma^2)\\
&= 1 -\gamma\sum\limits_{E's's}P(s)p_g(E')E'f(E's's)(f_{s'}-f_s) + \frac{\gamma}{2}\sum\limits_{E's's}P(s)p_G(E')(E'^2\gamma-1)(f_{s'}-f_s)^2) f(E's's)+\mathcal O(\gamma^2)
\end{align*}
To see that the second term on the right hand side is zero, note that for $\sum_{E's's}P(s)p_G(E')(E'^2\gamma-1)f_{s'}$ \eqref{stoc f 1} implies $\sum_s P(s) f(E's's)=P(s')$ $\forall$ $E'$, removing the $E'$ dependence from the $f(E's's)$. Then using $\langle E'^2\rangle =\gamma^{-1}$ sets this to zero. Similarly, with \eqref{stoc f 2} the term $-\sum_{E's's}P(s)p_G(E')(E'^2\gamma-1)f_{s}$ is identical to zero. Therefore we recover that 
\begin{equation}
\frac{\gamma}{2}\sum\limits_{E's's}P(s)p_G(E')(E'^2\gamma-1)(f_{s'}-f_s)^2) f(E's's) = \mathcal O (\gamma^2)
\end{equation}
\end{proof}

\noindent \textbf{Theorem 2} \emph{In any process, work is bounded form above by }
\[W \leq -\Delta F - \frac{1}{2\beta C}\Delta S^2 \ .
\] \label{quadratic upper bound}

\begin{proof}
Starting with the equation for the optimal work derived in Lemma \ref{optimal work equality}
\begin{equation}
W = -\Delta F - \underbrace{\frac{\gamma^2}{2\beta}\sum\limits_{E's'}P(s')p_G(E')E'^2 \left(  \frac{\sum\limits_s P(s) f(E's's)(f_{s'}-f_s)}{P(s')} \right)^2}_{[A]} + \underbrace{\frac{\gamma}{2\beta}\sum\limits_{E's's}p_G(E')(E'^2\gamma -1)P(s)f(E's's)(f_{s'}-f_s)^2}_{[B]}
\end{equation}
the upper bound can be derived through simple convexity arguments. First consider term $[A]$. 
\begin{align*}
[A] &= - \frac{\gamma^2}{2\beta}\sum\limits_{E's'}P(s')p_G(E')E'^2 \left(  \frac{\sum\limits_s P(s) f(E's's)(f_{s'}-f_s)}{P(s')} \right)^2 \\
&\leq - \frac{\gamma}{2\beta}\sum\limits_{E'}p_G(E')E'^2 \left(  \frac{\sum\limits_{ss'} P(s) f(E's's)(f_{s'}-f_s)}{P(s')} \right)^2 \\
&= - \frac{\gamma}{2\beta^3}\Delta S^2
\end{align*}
where we have used \eqref{stoc f 1} and \eqref{stoc f 2} and the concavity of the function $g(x)=-x^2$. Finally, consider term $[B]$. By \eqref{stoc f 1} and \eqref{stoc f 2} the sum $\sum_{ss'} P(s) f(E's's)$ is convex positive for all $E'$, and therefore 
\begin{align*}
[B] &\leq \frac{\gamma}{2\beta}\sum\limits_{E'}p_G(E')(E'^2\gamma -1)(\sum\limits_{ss'}P(s)f(E's's)( f_{s'}-f_s))^2\\
&= \frac{\gamma}{2\beta^3}\sum\limits_{E'}p_G(E')(E'^2\gamma -1)\Delta S^2\\
&= 0
\end{align*}
where we have used \eqref{stoc f 1} and \eqref{stoc f 2}, the convexity of the function $g(x)=x^2$ and $\langle E'^2\rangle = \gamma^{-1}$. Using these upper bounds for terms [A] and [B], and $\gamma=\beta^2/C$, we construct the desired upper bound for the work associated with the thermal operation.  
\end{proof}

\subsection{Information Erasure and state formation}

\noindent \textbf{Theorem 3:} \emph{For the work extraction process $(\rho, H_S)\rightarrow (\frac{1}{d}\mathbb I, \mathbb I)$, where a system with state $\rho$ and Hamiltonian $\mathcal H_S$ is taken the thermal state of a trivial Hamiltonian, the upper bound to the work can be achieved} \[W= -\Delta F - \frac{1}{2\beta C}\Delta S^2 \] 

\emph{whereas for the reverse ``Landauer erasure'' process $(\rho, H_S)\leftarrow (\frac{1}{d}\mathbb I, \mathbb I)$, the optimal work that can be achieved is} \[W= -\Delta F - \frac{1}{2\beta C}\left(\Delta S^2 + \mathsf{Var}(\rho) \right) \] \emph{where $\mathsf{Var}(\rho)$ is the varentropy of the final state $\rho$.}

\begin{proof}
For the Landauer erasure process, note that $P(s')=1/d$ which is independent of $s'$. Simply choosing $f(E's's)=P(s')=1/d$ and substituting this into the equation derived in Lemma \ref{optimal work equality} gives 
\begin{align}
W&=-\Delta F - \frac{\gamma^2}{2\beta^3}\sum\limits_{E's'}P(s')p_G(E')E'^2 \left( \sum\limits_s P(s)(\log d + \log P(s)) \right)^2 + \frac{\gamma}{2\beta^3}\sum\limits_{E's's}p_G(E')(E'^2\gamma -1)P(s)(\log d + \log P(s))^2\\
&= -\Delta F - \frac{\gamma}{2\beta^3}\sum\limits_{E's'}P(s')p_G(E')E'^2 \left( \log d -S(\rho)) \right)^2+0 \\
&=  -\Delta F - \frac{1}{2\beta C}\Delta S^2
\end{align}
Saturating the bound in Theorem \ref{quadratic upper bound}. Note that erasure is carried out optimally with the same thermal operation as for the infinite bath, $t(E's'|Es)=P(s')\delta (E-E'-f_{s'}+f_s)$. For the reverse process, note that $f_{s'}-f_s= -\log d - \log P(s')$ which is independent of $s$. Using \eqref{stoc f 1} to give $\sum_s P(s) f(E's's)=P(s')$, the equality in Theorem \eqref{optimal work equality} becomes 
\begin{align*}
W &=-\Delta F - \frac{\gamma^2}{2\beta^3}\sum\limits_{E's'}P(s')p_G(E')E'^2 \left( \sum\limits_s \frac{P(s)f(E's's)(\log d + \log P(s'))}{P(s')} \right)^2 + \frac{\gamma}{2\beta^3}\sum\limits_{E's's}p_G(E')(E'^2\gamma -1)P(s)f(E's's)(\log d + \log P(s))^2\\
&= -\Delta F - \frac{\gamma^2}{2\beta^3}\sum\limits_{E's'}P(s')p_G(E')E'^2 \left( \log^2 d + \log^2 P(s') + 2 \log d \log P(s')  \right)+0\\
&= -\Delta F -\frac{\gamma}{2\beta C}\left(\left( \log d - S(\rho) \right)^2 +\sum\limits_{s'}P(s') \log^2 P(s') - \left(\sum\limits_{s'}P(s')\log P(s')\right)^2 \right)\\
&=  -\Delta F -\frac{\gamma}{2\beta C}\left(\Delta S^2 + \mathsf{Var}(\rho) \right)
\end{align*}
where $\mathsf{Var}(\rho)$ is the variance of the surprise $-\log P(s')$ of the state that is distilled. 
\end{proof}

\section{Thermal operations with deterministic work}\label{app:det}

In this section we address the issue of extracting and expending deterministic work, free from fluctuations. 
As a first step, we shall describe what is the power of thermal operations with a finite bath when no work is involved, and only system-bath interactions occur. Then we can consider work extraction and expenditure using the same tools, by considering the weight as part of our system. This is now possible because we do not allow for arbitrary statistical fluctuations in our weight.

We will be making the approximation that the energy scales of the system are small compared to the energy fluctuations of the bath. This is accurate provided the bath is sufficiently large.

In the present context, out of the assumptions in Section \ref{sec:TOfluc}, only the first two are relevant. The energy conservation constraint implies that the unitaries we can apply follow

\begin{equation}\label{eq:encon}
[U,H_S+H_B]=0,
\end{equation}
 as the weight is not involved for now.

We append a system with a Hamiltonian $H=\sum_s \epsilon_{s} \ket{s}\bra{s}$ and state diagonal in energy $\rho=\sum_s p(s) \ket{s}\bra{s}$  to the finite bath, which again we consider to be in a thermal state $\tau_\beta$. Because of the form of \eqref{eq:encon}, we want to look at subspaces of fixed total energy of $\rho \otimes \tau_\beta$, as it was done in \cite{horodecki2013fundamental}, as the dynamics of each such subspace is then independent from the rest. Let $\Pi_{E_{\text{tot}}}$ be a projector into a total energy subspace of energy $E_{\text{tot}}$, any such subspace can be written as

\begin{align}\label{eq:energyproj}
&\Pi_{E_{\text{tot}}} \rho \otimes \tau_\beta \Pi_{E_{\text{tot}}}= \\ & \nonumber k_{E_{\text{tot}}} \sum_s \sum_{b\in \mathcal M({E_{\text{tot}}}-\epsilon_{s})} p(s) e^{\beta(\epsilon_{s}-{E_{\text{tot}}})} \ket{s}\bra{s} \otimes \ket{b}\bra{b},
\end{align}
where $k_{E_{\text{tot}}}$ is a normalization constant, and $\mathcal M (E)$ is the set of eigenstates of the bath with energy $E$, which has cardinality $\Omega(E)$. In contrast to this, looking at each particular total energy subspace is not necessary in the case of fluctuating work because one can average over it (or more precisely, over $E$ and $\epsilon_{s}$, the energies of bath and system) in the final results. This is explicit for instance in Eq. \eqref{eq:avgwork}.

We assume that $\epsilon_{s}^2 << 1/\gamma$ for all $s$. Effectively this means we can approximate the density of states from Eq. \eqref{eq:densityapp} as

\begin{equation} \label{eq:bsapprox}
 \Omega({E_{\text{tot}}} - \epsilon_{s}) \simeq \Omega ({E_{\text{tot}}}) e^{- (\beta - \gamma {E_{\text{tot}}})\epsilon_{s}}.
\end{equation} 
This gives the size of the subspace with the same ${E_{\text{tot}}},\epsilon_{s}$ that appears in the expansion \eqref{eq:energyproj}. In each of these subspaces the projectors have weights $\propto k_{E_{\text{tot}}} e^{-\beta {E_{\text{tot}}}} p(s) e^{\beta \epsilon_{s}}$. We can order these in decreasing magnitude as $p(1) e^{\beta \epsilon_1} \ge p(2) e^{\beta \epsilon_2} \ge ....\ge p(n) e^{\beta \epsilon_n}$, which is usually refered to as $\beta$-order. This order allows us to draw a majorization diagram for the subspace. An example is shown in Fig. \ref{fig:Maj}.

Now we look at the effect of an arbitrary energy-conserving unitary $U$ to the joint system-bath pair. This has been shown to be equivalent to applying an arbitrary mixture of unitaries in each given total energy subspace \cite{horodecki2013fundamental}. By the Birkhoff-von Neumann theorem, such an arbitrary mixture of unitaries is equivalent to the full set of bistochastic maps, and hence a transformation is possible if the probability distribution of each energy subspace of the initial state \emph{majorizes} that of the energy subspace with that same energy of the final state \cite{marshall2010inequalities}. Given two probability distributions $\lambda(s)$ and $\eta(s)$ with $n$ elements, ordered such that $\lambda(s+1)\ge \lambda(s)$ and $\eta(s+1)\ge \eta(s)$, we say $\lambda(s)$ majorizes $\eta(s)$ if it is true that 
\begin{equation}
\sum_{s=1}^k \lambda(s) \ge \sum_{s=1}^k \eta(s)  \,\,\, \forall k \in \{1,n\}.
\end{equation}

A transition between two different states of system and bath, $\rho \otimes \tau_\beta$ and $\sigma\otimes \tau_\beta$ is hence possible if and only if in each of the total energies $E$ the distribution of $\rho\otimes \tau_\beta$, as seen in the majorization diagram, majorizes that of $\sigma\otimes \tau_\beta$. This is better phrased in terms of \emph{thermomajorization diagrams}. We may define the concave 2D diagram, with origin at $\{0,0\}$, as the curve resulting from consecutively joining the points given by 

\begin{equation}
\Big \{ \frac{\sum_s^k  p(s) e^{ \gamma {E_{\text{tot}}} \epsilon_{s}}}{\sum_s^n  p(s) e^{ \gamma {E_{\text{tot}}} \epsilon_{s}}},\sum^k_s e^{-(\beta -  \gamma {E_{\text{tot}}}) \epsilon_{s}} \Big \}\,\,\, \forall k \in \{1,n\}
\end{equation}
 where the eigenstates are labeled according to the $\beta$-ordering. Each segment of this diagram corresponds to a particular energy eigenstate of the system, and is constructed by subsequently adding the probability weight of each of the individual ``microstates'' of the majorization diagram (the individual bars of Fig. \ref{fig:Maj}).
 
Majorization of two states of a particular subspace of total energy ${E_{\text{tot}}}$ is then equivalent to the inital thermomajorization diagram of energy ${E_{\text{tot}}}$ lying stricly above that of the target one. Hence, we have that for each energy ${E_{\text{tot}}}$ there is a slightly different thermomajorization-like criteria, and for a transition to be possible with complete certainty all of them apply. An example of such a diagram is given in Fig. \ref{fig:ThermoMaj}. 

\begin{figure}
\includegraphics[width=0.5\textwidth]{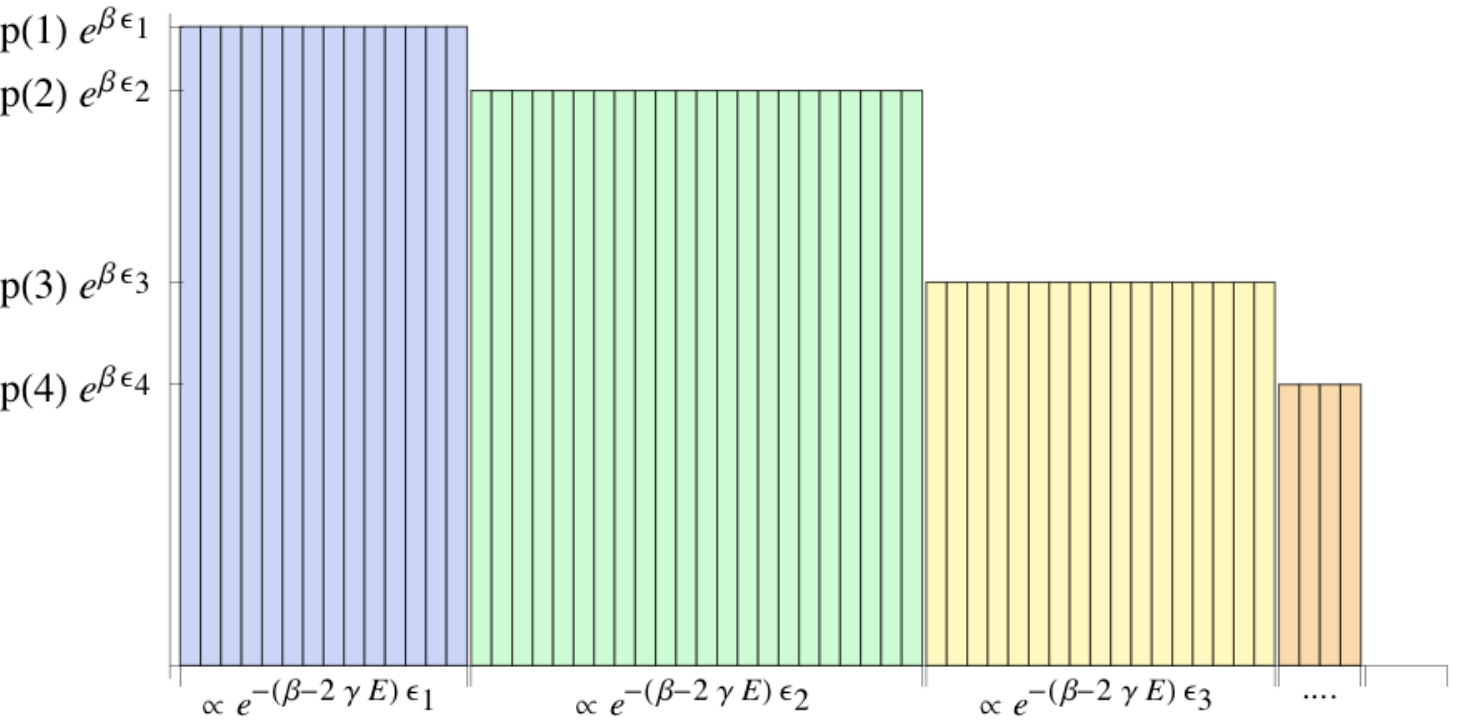}
\caption  {We show the majorization diagram of a particular subspace of total energy $E$ of the joint system-bath product state, where the system is in state $\rho=\sum_s p(s) \ket{s}\bra{s}$ diagonal in energy, and with the energy levels $\beta$-ordered such that $p(k)e^{\beta E_k} \ge p(k+1)e^{\beta E_{k+1}}$.} \label{fig:Maj}
\end{figure}

At ${E_{\text{tot}}}=0$, the average energy, this criteria is equivalent to the standard thermomajorization of \cite{horodecki2013fundamental}, showing that the conditions here are strictly stronger. Indeed, standard thermomajorization is the necessary and sufficient criteria for checking whether a transition is possible in the case of an infinite bath. Equivalently, this means that control over such an infinite bath gives one the power to implement any Gibbs-preserving stochastic matrix on a state \cite{marshall2010inequalities,renes2014work}. With the restriction proposed here, the operations still amount to a set of Gibbs-preserving stochastic matrices (as the thermal state is still the fixed point), but not any such stochastic matrix can be generated with a finite bath only. For other explicit examples of similar limitations, given a different model of a finite bath, see \cite{scharlau2016quantum}.

\begin{figure}
\includegraphics[width=0.5\textwidth]{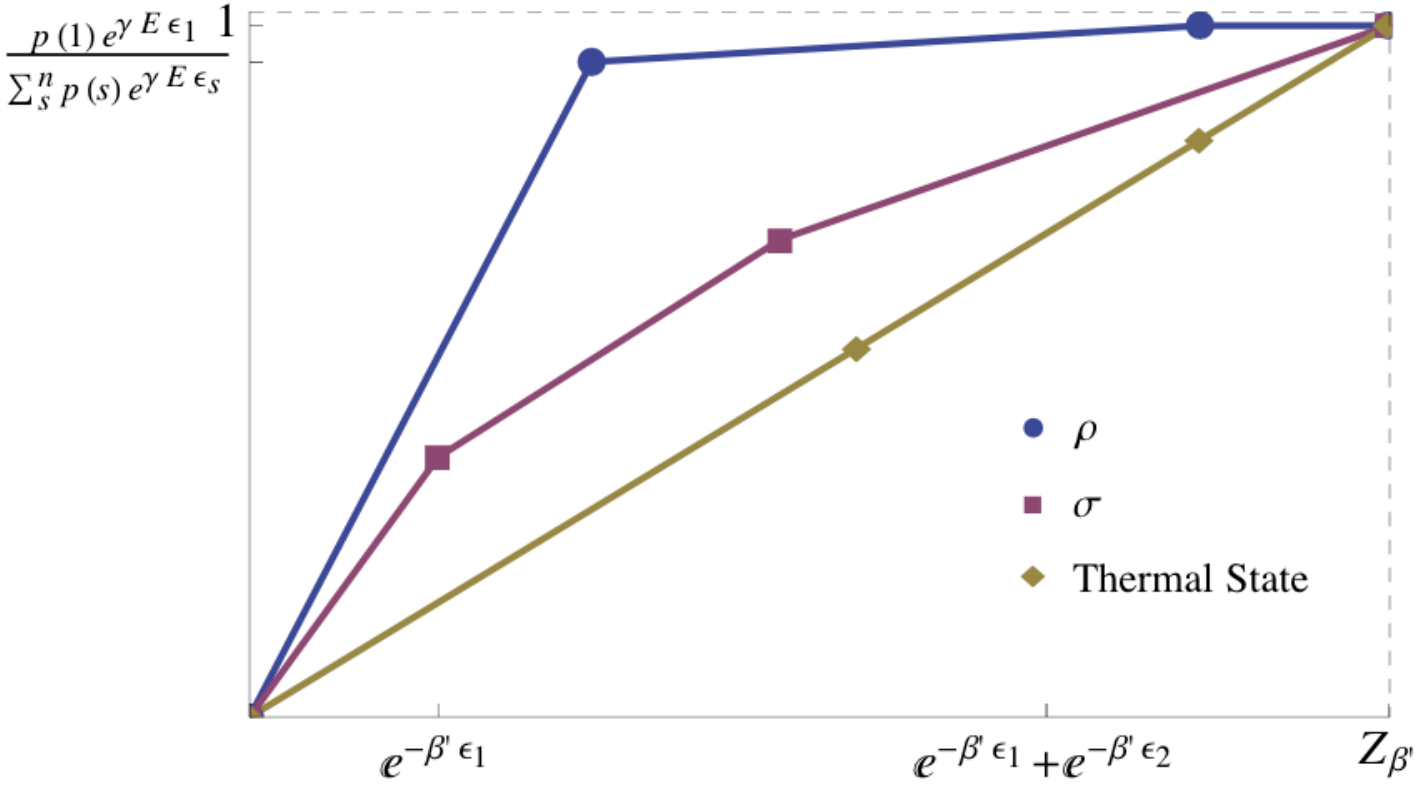}
\caption  {We show an example of a thermomajorization diagram for three different states of a $3$-level system, $\rho$, $\sigma$ and the thermal state, or a total energy of $E$. The $\beta$-order in each is different, and at this particular energy the transitions $\rho \rightarrow \sigma$ is possible, as the diagram of $\rho$ lies above that of $\sigma$.} \label{fig:ThermoMaj}
\end{figure}

Each diagram is given by a particular energy, and to give a conclusive answer we have to check within the whole range of possible $E_\text{tot}$, which may not be feasible in general, as the range of $E_\text{tot}$ is continuous. However, it may be possible in particularly easy cases, such as qubits, where in many cases, such as where the $\beta$-order is the same initially and finally, only the initial slope has to be compared (and this is independent of the particular total energy), or for systems for trivial Hamiltonians, where one recovers standard majorization. These two facts are consistent with the results of \cite{scharlau2016quantum}, where it is shown that \emph{i)} a large class of maps can be applied to qubits using a finite bath and \emph{ii)} for trivial Hamiltonians, or Noisy Operations, one only needs a bath of the same size as the system. 

In any case, considering the full range of total energies may not be necessary in particular cases. We know from Section \ref{sec:finite} that the energy distribution Eq. \eqref{eq:densityapp} is a Gaussian with tails that cause the range of energies we have to check to be very large. However, the further the energy is from the mean, the less likely it is to occur. Hence, this criteria may give a definite answer only if we are willing to make the transition with probability slightly less than 1, and just consider a certain range of total energies. The next subsection is devoted to analyzing the tradeoff between the range of energies that need to be checked and the probability of failure.

\subsection{Total energy distribution}\label{app:dist}

Given that we have the energy distributions of both bath and system, we can obtain the total energy distribution.
\begin{equation} 
p(E_{\text{tot}})\propto \sum_s e^{-\frac{\gamma}{2}(E_{tot}^2- E_{tot} \epsilon_{s})} p(s),
\end{equation}
which is a convex mixture of Gaussians with width $1/\sqrt{\gamma}$, each with its centre offset by $\epsilon_{s}$. If we make the slightly stronger approximation than in Eq. \eqref{eq:bsapprox} above that $1/\sqrt{\gamma}>>\epsilon_{s}$, then we can approximate the total energy distribution to be that of just the bath, so $E_{\text{tot}} \sim E \sim \gamma^{-1/2}$ and 
\begin{equation}
p(E_{\text{tot}}) \approx p_b(E) \propto e^{-\frac{\gamma}{2} E^2}.
\end{equation}
Note that the previous approximation, in Eq. \eqref{eq:bsapprox}, was only $1/\gamma>>\epsilon^2_s$.

Now say we want a transformation to happen with probability $1-\epsilon$. The range of energies $\{ E_{\text{min}},E_{\text{max}} \}$ that have to be checked is then given by the integral equation

\begin{equation}\label{eq:eps}
1-\epsilon = \int_{E_{\text{min}}}^{E_{\text{max}}} \sqrt{\frac{\gamma}{2\pi}} e^{-\frac{\gamma}{2} E^2} \text{d} E.
\end{equation}

We can also define two minimum and maximum effective temperatures given by
\begin{align} \label{eq:bmax}
\beta_{+}= \beta- \gamma E_{\text{max}} \\ \label{eq:bmin}
\beta_{-}=\beta- \gamma E_{\text{min}}
\end{align}

We may take $-E_{\text{min}}=E_{\text{max}}=E^*$, then
\begin{equation}
\epsilon = 2 \sqrt{\frac{\gamma}{\pi}} \int_{E^*}^{\infty} e^{-\gamma E^2} \text{d} E.
\end{equation}

For large $x$ we can approximate the error function as 
\begin{equation}
\text{errf}(x) \simeq 1-\frac{e^{-x^2}}{\sqrt{\pi} x}.
\end{equation}

Hence we can approximate
\begin{equation}
\epsilon \simeq 2^{3/2} \frac{e^{-\frac{\gamma}{2} {E^*}^2}}{\sqrt{\pi \gamma} E^*}
\end{equation}
This expression gives the tradeoff between the range of thermomajorization diagrams we need to check, which will be $E \in \{ -E^* , E^*\}$, and the probability of failing for the criteria being definite.

\subsection{Deterministic work extraction and expenditure} \label{app:detwork}

We are now in a position to calculate the deterministic work in certain transitions, given a probability of error $\epsilon$. We specialize to two important cases where the answer can be written in a closed form: extracting work in a transition to a thermal state (as the least resourceful state), and creating an arbitrary state out of a thermal state.

\subsubsection{Work extraction in total energy subspaces}

From a given initial state $\rho$ we want to extract some deterministic work $W$ given some total energy subspace $E$. We define this to be equivalent to computing what is the maximum $W$ for which the transition
\begin{equation}\label{eq:trans1}
\rho \otimes \ket{0}\bra{0} \rightarrow \tau \otimes \ket{W}\bra{W},
\end{equation}
is possible given total energy $E$. Here, $\tau=e^{-\beta H_S}/Z_S$ is the thermal state of the system, and the second system of the tensor product is the weight. From now we can ignore all the energy levels of the weight except for the two involved in the transition, so that we can effectively think of it as a two-level system with an energy gap between the two given by $W$ (and hence, a Hamiltonian $H_W=\ketbra{W}{W}$). Due to the lack of fluctuations in the weight, we can consider it together with the working system, and hence we can use the results of the previous subsection, which were phrased for the system alone.

The curve of the final state of Eq. \eqref{eq:trans1} is just a straight line, with slope given by $e^{-\beta W}/Z_S$. Hence, in order to check whether the transition is possible we just need to compare the ranks (or the points at which the thermomajorization diagrams reach $1$ in the vertical axis) of the subspace of energy $E$ of the states $\rho \otimes \ket{0}\bra{0} \otimes \tau_\beta$ with $\tau \otimes \ket{W}\bra{W} \otimes \tau_\beta$ \cite{horodecki2013fundamental,gemmer2015single}. 
The former is
\begin{equation}
d_\text{in}=\sum_s \Omega (E-\epsilon_{s}) p(s)^0 \simeq \Omega(E) \sum^n_s e^{-(\beta - \gamma E) \epsilon_{s}} p_s^0,
\end{equation}
and the latter is
\begin{equation}
d_\text{fin}=\sum_s \Omega (E-\epsilon_{s}-W) \simeq \Omega(E) e^{-(\beta - \gamma E) W} \sum^n_s e^{-(\beta - \gamma E) \epsilon_{s}},
\end{equation}
where we have approximated the densities of states the same way as \eqref{eq:bsapprox}.
In general, we'll have that $d_\text{in}\le d_\text{fin}$, and both ranks will be equal in the optimal situation where the maximum work $W$ is extracted. Equating and solving for $W$ yields
\begin{equation}\label{eq:minW}
W= \frac{1}{\beta'} F^{\beta'}_{\text{min}}(\rho),
\end{equation}
where $\beta'=\beta-\gamma E$. This is hence just the $0$-order Renyi free energy with the temperature fluctuating. 

The quantity of Eq. \eqref{eq:minW} is the maximum amount of work one can extract for a given energy $E$. If we want to extract truly deterministic work, the work extracted must however be the same independently of what is the actual total energy. The minimum possible $W$ over all $E$ hence gives the highest possible one that can be the same for all energies. The function $\frac{1}{\beta'} F^{\beta'}_{\text{min}}(\sigma)$ is monotonically increasing in $E$, so the minimum is achieved at our chosen $E=E_{\text{min}}$, or in $\beta'=\beta_{-}$ as defined in Eq. \eqref{eq:bmin}. The optimal deterministic extractable work is then
\begin{equation}\label{eq:minW2}
W_{\text{ext}}= \frac{1}{\beta_{-}} F^{\beta_{-}}_{\text{min}}(\rho),
\end{equation}
with a failure probability $\epsilon$ determined by \eqref{eq:eps}.

\subsubsection{Work of formation in total energy subspaces}
Now, for a given total energy $E$, we want the lowest $W$ for which the transition 
\begin{equation}
\tau \otimes \ket{W}\bra{W} \rightarrow \rho \otimes \ket{0}\bra{0} 
\end{equation}
is possible. Because the curve of the initial state is straight, we only need its slope to be equal to the slope of the first segment (the one corresponding to $E_1$, $\beta$-ordered) of the final one \cite{horodecki2013fundamental,gemmer2015single}. 
The slope of the initial state is
\begin{equation}
\frac{e^{\beta' W}}{Z_{\beta'}},
\end{equation}
and the final one is

\begin{equation}
\frac{1}{\sum_{l}^n  p(l) e^{ \gamma E \epsilon'_{l}}} \text{max}_s p(s) e^{\beta \epsilon_{s}}
\end{equation}

Making them equal and solving for $W$ yields
\begin{equation}
W=\frac{1}{\beta'} (\log{ (Z_{\beta'} \text{max}_s p(s) e^{\beta \epsilon_{s}})} +\log{\frac{1}{\sum_{s'}^n  p(s) e^{\gamma E E_j}}}).
\end{equation}
We can decompose this in three terms:
\begin{equation}\label{eq:wform}
W=\frac{1}{\beta'}(\log{\frac{Z_{\beta'}}{Z_\beta}}+F_{\max}^\beta(\rho)+ \log{\frac{1}{\sum_{s}^n  p(s) e^{ \gamma E \epsilon_{s}}}})
\end{equation}
One is a change of partition function with respect to the infinite-size partition function, the other one is the $\infty$-order Renyi free energy at infinite size, and the last one is a term that is always positive, and converges to zero in the limit of infinite size. Unlike the previous case, this is not equal to the expression in the infinite bath limit with a corrected temperature.

Analogously to the extractable work, if we want to create a state with deterministic work we need to find the maximum of these over $E$. The work $W$ in Eq. \eqref{eq:wform} is also monotonically increasing in $E$, and hence the maximum is achieved at $E=E_{\text{max}}$, or equivalently in $\beta'=\beta_{+}$ as defined in Eq. \eqref{eq:bmax}. The minimum deterministic work we need is hence
\begin{equation}
W_{\text{for}}=\frac{1}{\beta_{+}}(\log{\frac{Z_{\beta_{+}}}{Z_\beta}}+F_{\max}^\beta(\rho)+ \log{\frac{1}{\sum_{s}^n  p(s) e^{\gamma E_{\text{max}} \epsilon_{s}}}})
\end{equation}

\begin{figure}
\includegraphics[width=0.5\textwidth]{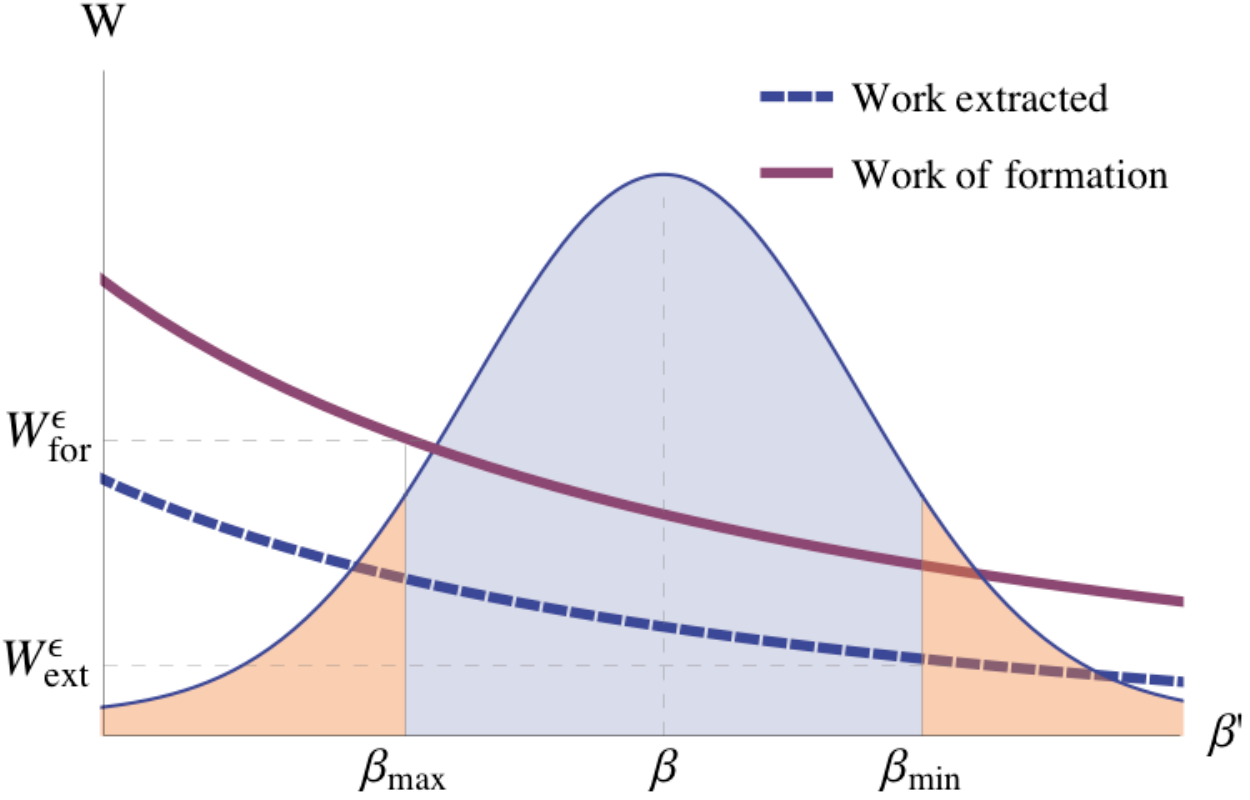}
\caption  {We show the curves for the maximum work extracted and minimum work of formation for a given state $\rho$, given each particular total energy. The distribution of total energies (and hence the distribution of $\beta'=\beta-\gamma E$) is the superposed Gaussian curve, for which a cut in the tails gives the values of effective temperatures $\beta_{-}$ and $\beta_{+}$ that determine the work of formation and the extractable work given some error probability, whose value is equal to the orange region.} \label{fig:EpsW}
\end{figure}

In Fig. \ref{fig:EpsW} we show the maximum extractable work and the minimum work of formation for a particular example, as a function of the modified temperature $\beta'$.

\end{document}